\documentclass[conference]{IEEEtran}
\IEEEoverridecommandlockouts

\usepackage{amsmath,amsfonts}
\usepackage{graphicx}
\usepackage{textcomp}
\usepackage{xcolor}
\usepackage{bbding}
\usepackage{bbm}
\usepackage{bm}
\usepackage{booktabs} 
\usepackage{pifont}
\usepackage{balance}
\usepackage{color}
\usepackage{ifpdf}
\usepackage{latexsym}
\usepackage{paralist}
\usepackage{comment}
\usepackage{xspace}
\usepackage{mathrsfs}
\usepackage{epsf}
\usepackage{setspace}
\usepackage{caption}
\usepackage{mathtools}
\usepackage{multirow}
\usepackage{url,epsfig,array}
\usepackage{leftidx}
\usepackage{epstopdf}
\usepackage{footmisc}
\usepackage{float}
\usepackage{amsthm}
\usepackage{enumitem}
\usepackage{subcaption}
\usepackage{xspace}
\usepackage[table]{xcolor}
\usepackage[noend]{algpseudocode}
\usepackage[ruled,linesnumbered,vlined]{algorithm2e}

\usepackage{marvosym}
\def\ie{\textit{i.e.}\xspace}

\def\eg{\textit{e.g.}\xspace}

\newtheorem{corollary}{Corollary}

\newtheorem{definition}{Definition}
\newtheorem{proposition}{Proposition} 

\def\BibTeX{{\rm B\kern-.05em{\sc i\kern-.025em b}\kern-.08em
    T\kern-.1667em\lower.7ex\hbox{E}\kern-.125emX}}
\begin{document}

\title{NestPipe: Large-Scale Recommendation Training on 1,500+ Accelerators via Nested Pipelining}

\author{\IEEEauthorblockN{Zhida Jiang\textsuperscript{1*}, Zhaolong Xing\textsuperscript{1*}\thanks{\textsuperscript{*}Equal Contribution.}, Huichao Chai\textsuperscript{2}, Tianxing Sun\textsuperscript{1}, Qiang Peng\textsuperscript{1}, Baopeng Yuan\textsuperscript{1}, Jiaxing Wang\textsuperscript{1},
\\Hua Du\textsuperscript{1}, Zhixin Wu\textsuperscript{2}, Xuemiao Li\textsuperscript{2}, Yikui Cao\textsuperscript{2}, Xinyu Liu\textsuperscript{2}, Yongxiang Feng\textsuperscript{2}, Zhen Chen\textsuperscript{1\Letter}\thanks{\textsuperscript{\Letter}Corresponding Author.}, Ke Zhang\textsuperscript{1}}
	\IEEEauthorblockA{
		\textsuperscript{1}\textit{JD.com} \quad\quad\quad
		\textsuperscript{2}\textit{Huawei}
} 
}


\maketitle
\begin{abstract}

Modern recommendation models have increased to trillions of parameters. As cluster scales expand to \textit{O(1k)}, distributed training bottlenecks shift from computation and memory to data movement, especially lookup and communication latency associated with embeddings.
Existing solutions either optimize only one bottleneck or improve throughput by sacrificing training consistency. This paper presents NestPipe, a large-scale decentralized embedding training framework that tackles both bottlenecks while preserving synchronous training semantics.
NestPipe exploits two hierarchical sparse parallelism opportunities through nested pipelining.
At the inter-batch level, Dual-Buffer Pipelining (DBP) constructs a staleness-free five-stage pipeline through dual-buffer synchronization, mitigating lookup bottlenecks without embedding staleness.
At the intra-batch level, we identify the embedding freezing phenomenon, which inspires Frozen-Window Pipelining (FWP) to overlap All2All communication with dense computation via coordinated stream scheduling and key-centric sample clustering. 
Experiments on production GPU and NPU clusters with 1,536 workers demonstrate that NestPipe achieves up to 3.06$\times$ speedup and 94.07\% scaling efficiency.
\end{abstract}


\section{Introduction}\label{sec:intro}
Recent advances in the recommendation domain have validated the scaling law similar to that of large language models (LLMs) \cite{zhang2024scaling,han2025mtgr,xu2025climber}.
Scaling up model parameters and training data consistently yields substantial improvements in recommendation quality.
This trend has driven the rapid development of next-generation large-scale recommendation models, which may scale to trillions of parameters \cite{ding2026bending}. 
Despite the evolution of recommendation architectures, sparse embedding tables remain a critical component, enabling effective representation learning of user behaviors and item characteristics \cite{qiu2025evolution,lai2023adaembed,cheng2026conditional}.
In practice, embedding tables dominate the overall parameter footprint of recommendation models and occupy terabytes of memory. 
To enable efficient training of large-scale recommendation models, industrial training systems have expanded to distributed clusters comprising thousands of accelerators (\ie, workers), creating an urgent demand for scalable parallelization strategies tailored to sparse workloads \cite{zhang2025two}.
These massive embedding tables are partitioned across multiple workers via model parallelism \cite{yang2025research}, and each worker leverages a hierarchical storage architecture for expanding available memory beyond HBM constraints \cite{kurniawan2023evstore,zhao2020distributed}.

However, such decentralized training architecture faces a scalability paradox as the cluster scales to thousands of workers (\ie, \textit{O(1k)} and beyond).
While underlying computational power grows significantly, actual training efficiency fails to keep pace as expected.
We highlight that the scalability barrier of sparse training shifts from memory and computation to data movement, especially the \textit{lookup} and \textit{communication} associated with sparse operations.
Firstly, the lookup operations encompass data preprocessing, distributed key routing, embedding retrieval, and host-to-device (H2D) transfers. This overhead, which is negligible in small-scale scenarios, creates severe blocking with increasing batch sizes and sequence lengths.
Secondly, due to model parallelism, the exchange of embedding vectors and their corresponding gradients relies on All2All collective communication. 
Since All2All requires fully connected peer-to-peer data exchange with quadratic connection complexity, communication latency grows super-linearly with cluster scale, hindering efficient embedding training even with state-of-the-art (SOTA) high-speed interconnects.

Although existing asynchronous sparse training schemes \cite{lian2022persia,miao2025efficient,su2022gba,huang2021hierarchical} aim to hide lookup or communication latency, they typically cause parameter staleness and relax the consistency guarantees required by production training. 
Similarly, embedding compression methods \cite{lai2023adaembed,liu2025cafe+,feng2024accelerating,wang2024accelerating} can reduce the volume of transferred data but inevitably introduce information loss.
These solutions often sacrifice training consistency for throughput gains, which may ultimately damage the convergence stability, especially for generative recommendation models that are highly sensitive to parameter consistency \cite{hou2025generative}.
More importantly, most studies are designed for small-scale training configurations and fail to resolve the latency exposed in large-scale industrial deployments, thus lacking scalability.
The most recent SOTA optimization for large-scale embedding training is two-dimensional sparse parallelism \cite{zhang2025two}, which restricts All2All communication within local groups via intra-group model parallelism and inter-group data parallelism. Unfortunately, such topology-oriented paradigms alter parameter update logic and cause potential accuracy loss. Besides, it operates in a traditional synchronous manner, which leaves expensive computing resources idle during communication and thus degrades hardware utilization \cite{li2025slimpipe}.
As a result, existing works are trapped in the accuracy-throughput dilemma and suffer from severe scalability degradation as the cluster scale expands.
This paper argues that the core issue behind accuracy-throughput dilemma is that we always attempt to reduce the absolute lookup or communication overhead, not the \emph{exposed} portion of that overhead on the critical workflow of synchronous training. 
From this perspective, we propose NestPipe, a large-scale decentralized embedding training framework with nested pipelining, which exploits two sparse parallelism opportunities at different spatial granularities, \ie, an inter-batch pipeline window for lookup and an intra-batch frozen window for communication.
Specifically, NestPipe incorporates the Dual-Buffer Pipelining (DBP) strategy, which addresses lookup latency by constructing a staleness-free five-stage pipeline.
DBP strategy preserving embedding freshness in pipelining by forcing precise synchronization points and alternating buffer usage.

Besides, we revisit the updating frequency of sparse embedding vectors and identify the parameter freezing phenomenon to implement the Frozen-Window Pipelining (FWP) strategy. 
Rather than devoting extensive efforts to keeping remote embeddings up-to-date, FWP strategy shifts the optimization perspective from macro to micro and hides All2All communication behind dense computation within a semantically valid frozen window.
At the implementation level, FWP realizes this finer-grained overlap through coordinated computation and communication stream scheduling. 
Sparse communication is launched early within the frozen window and dense computation proceeds on ready micro-batches, thereby maximizing overall resource utilization.
To further balance overlap opportunity and raw communication, FWP incorporates lightweight key-centric sample clustering to improve key deduplication across micro-batches, which reduces repeated embedding transmission and approaches the theoretical exposed ratio.

The combination of DBP and FWP yields a hierarchical sparse parallelism design that addresses the two efficiency bottlenecks introduced by large-scale embedding training. 
We also provide theoretical consistency analysis for NestPipe.
As a result, our solution remains \textit{efficient} in reducing both lookup and communication latency, \textit{consistent} with standard synchronous training semantics, and \textit{scalable} as cluster scale grows.
More importantly, NestPipe optimizes exposed ratio rather than absolute overhead, which makes it naturally orthogonal to existing embedding sharding \cite{zha2023pre,liu2024embedding,zha2022autoshard,wang2024oper,mudigere2022software,zeng2024accelerating}, compression \cite{lai2023adaembed,liu2025cafe+,li2024mixed,huang2026fusedrec,feng2024accelerating,zhou2024dqrm,wang2022rec,wang2024accelerating}, and communication topology optimization \cite{zhang2025two,luo2024disaggregated}.
The main contributions of this paper are summarized as follows:
\begin{itemize}[leftmargin=*]
    \item We propose NestPipe, a decentralized framework that exploits two hierarchical sparse parallelism granularities via nested pipelining.
    NestPipe addresses both lookup and communication bottlenecks exposed by large-scale embedding training while maintaining synchronous training semantics.
    \item At the inter-batch level, DBP strategy hides lookup latency by a staleness-free five-stage pipeline. DBP synchronizes the intersection of active and prefetch buffers to eliminate embedding staleness without breaking pipeline parallelism.
    \item At the intra-batch level, FWP strategy identifies the parameter freezing phenomenon to overlap All2All communication with dense computation via coordinated stream scheduling. FWP further employs key-centric sample clustering to balance overlap opportunity and raw communication.
    \item We implement NestPipe on industrial-grade GPU and NPU clusters ranging from 128 to 1,536 workers.
    Experimental results on different workload characteristics show that NestPipe achieves up to 3.06$\times$ speedup and maintains 94.07\% scaling efficiency compared with SOTA baselines. Integrating with two-dimensional sparse parallelism, the speedup and scalability are further improved to 3.18$\times$ and 97.17\%.
\end{itemize}

\section{Background and Motivation}\label{sec:background}

\subsection{Hybrid Decentralized Training Architecture}

Recent industrial recommender systems often contain trillions of parameters dominated by sparse embeddings, complemented by complex dense layers for computation. 
This expansion is further encouraged by recent research demonstrating that scaling model parameters yields substantial performance gains, ultimately improving return on investment \cite{zhai2024actions,han2025mtgr,ding2026bending}.
To efficiently train such large-scale recommendation models across thousands of workers, training systems often adopt a hybrid parallelism architecture that decouples sparse and dense parameter management \cite{zhao2020distributed,wang2022merlin,DBLP:journals/corr/abs-2509-20883}.
As shown in Fig. \ref{fig:hybrid arch}, sparse embedding tables are sharded across workers via model parallelism and stored in hierarchical host-device memory, while dense layers are replicated via data parallelism.
Each worker exclusively manages a subset of embedding shards.
Hierarchical storage architecture is typically employed to accommodate the embedding shards that exceed the HBM capacity of local workers \cite{kurniawan2023evstore,zhao2020distributed}.
The device memory (HBM) serves as a high-performance cache to maintain the embedding vectors for the current and upcoming training batches \cite{miao2025efficient,song2025unified}.

\begin{figure}
    \centering
    \includegraphics[width=0.5\textwidth]{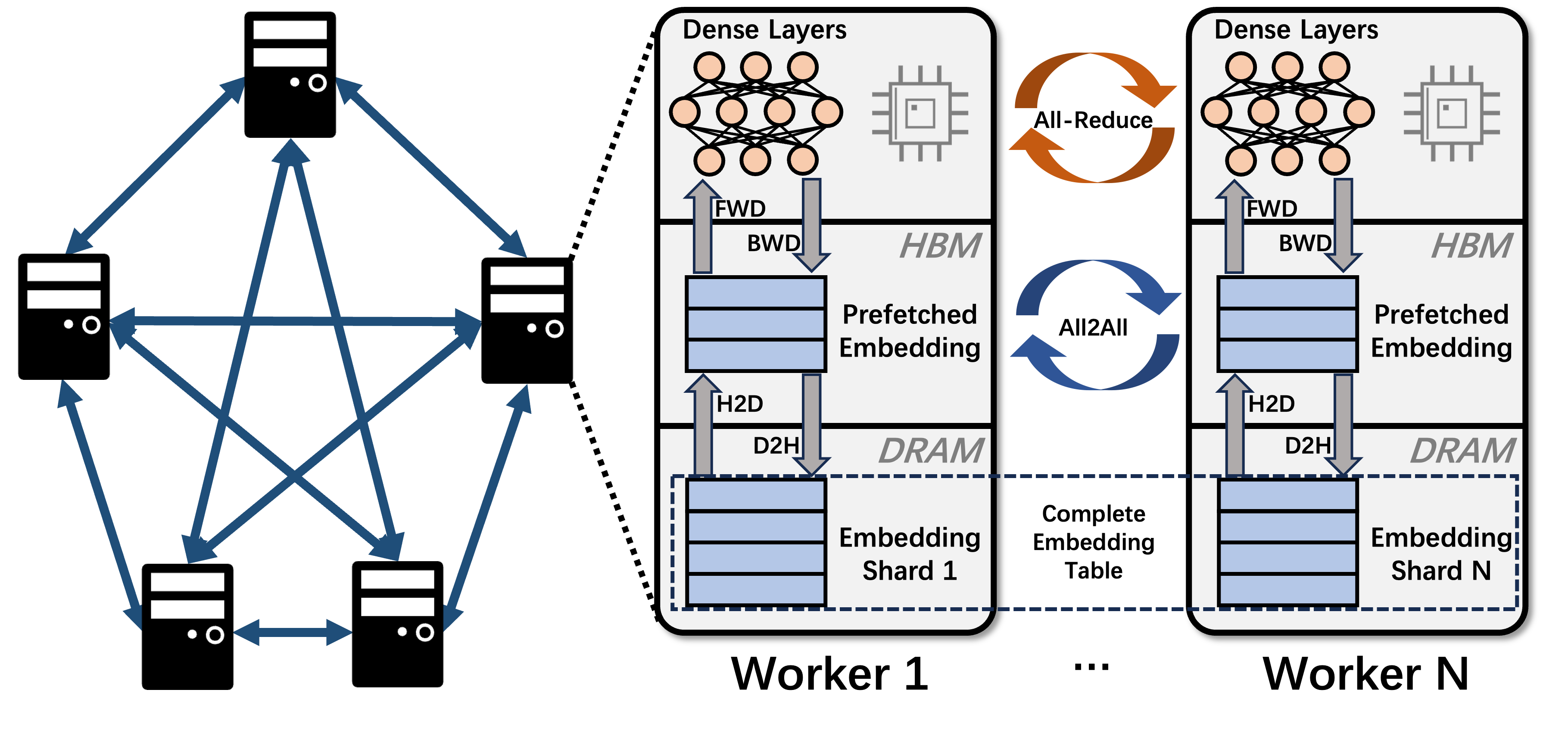} 
    \caption{Hybrid decentralized architecture for large-scale recommendation training. 
    }
    \label{fig:hybrid arch}
    \vspace{-2mm}
\end{figure}

\begin{table*}[t]
\centering
\caption{Comprehensive comparison of existing works with NestPipe.}
\label{tab:comparison}
\resizebox{0.95\textwidth}{!}{
\begin{tabular}{@{}ccc c c c ccc @{}}
\toprule
\multirow{2}{*}{\textbf{Different Methods}} & \multicolumn{2}{c}{\textbf{Efficiency}}  & \multirow{2}{*}{\textbf{Consistency}} &\multirow{2}{*}{\textbf{Scalability}} &\textbf{Orthogonality} \\ \cmidrule(lr){2-3}
& \textbf{Lookup} & \textbf{Communication} & &  & \textbf{with NestPipe}    \\ \midrule

Asynchronous Training (\eg, \cite{lian2022persia,miao2025efficient,su2022gba,huang2021hierarchical,zhang2022picasso,adnan2024heterogeneous,kwon2022training}) & \ding{51}  & \ding{55}       & \ding{55} & \ding{55}  & \ding{55}  \\ 
Embedding Compression (\eg, \cite{lai2023adaembed,liu2025cafe+,li2024mixed,huang2026fusedrec,feng2024accelerating,zhou2024dqrm,wang2022rec,wang2024accelerating}) &\ding{51}  & \ding{51}         & \ding{55} & \ding{55}  & \ding{51}  \\ 
Embedding Sharding and Scheduling (\eg, \cite{zha2023pre,liu2024embedding,zha2022autoshard,wang2024oper,mudigere2022software,zeng2024accelerating}) & \ding{55}  & \ding{51}       & \ding{51} & \ding{55}  & \ding{51}   \\ 
Two-dimensional Sparse Parallelism \cite{zhang2025two} & \ding{55}  & \ding{51}       & \ding{55} & \ding{51}  & \ding{51}   \\ 
\cellcolor{gray!20}\textbf{NestPipe (Ours)}& \cellcolor{gray!20}\ding{51}  & \cellcolor{gray!20}\ding{51}  &\cellcolor{gray!20}\ding{51} & \cellcolor{gray!20}\ding{51}    & \cellcolor{gray!20}\textbf{-}    \\ 
\bottomrule
\end{tabular}}
\vspace{-2mm}
\end{table*}

During the forward propagation of each training step, categorical features in the current batch correspond to embeddings residing on different workers.
Based on the exchanged keys, the workers retrieve the corresponding embedding vectors from host DRAM and then perform the H2D data transfer.
Based on a scalable peer-to-peer topology, the workers exchange the retrieved embedding vectors residing on HBM via All2All communication, ensuring that each worker collects all the embeddings required for the current batch. 
After dense computation, the gradients are routed back to their owner workers via All2All communication and aggregated to update the corresponding embedding vectors.
Finally, the embedding vectors are written back to host DRAM to maintain the consistency of complete embedding table in the hierarchical storage.
In comparison, dense layers are orders of magnitude smaller in size, and their gradients are synchronized through All-Reduce collective communication \cite{he2025hypereca}.

\begin{figure}
    \centering
    \begin{subfigure}{0.49\columnwidth}
        \centering
        \includegraphics[width=\linewidth]{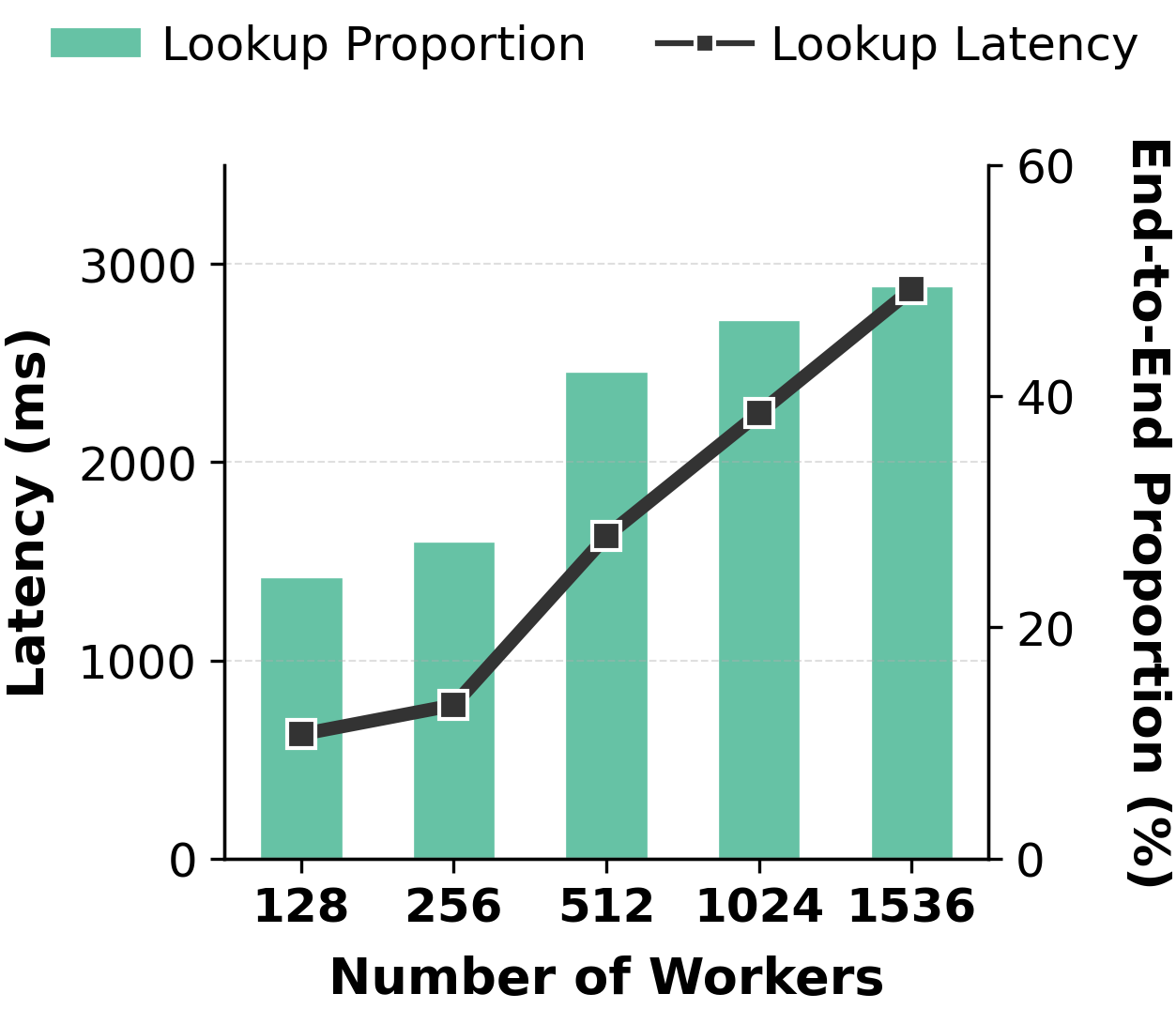}
        \caption{Lookup Overhead}
        \label{fig:io_overhead}
    \end{subfigure}
    \hfill 
    \begin{subfigure}{0.49\columnwidth}
        \centering
        \includegraphics[width=\linewidth]{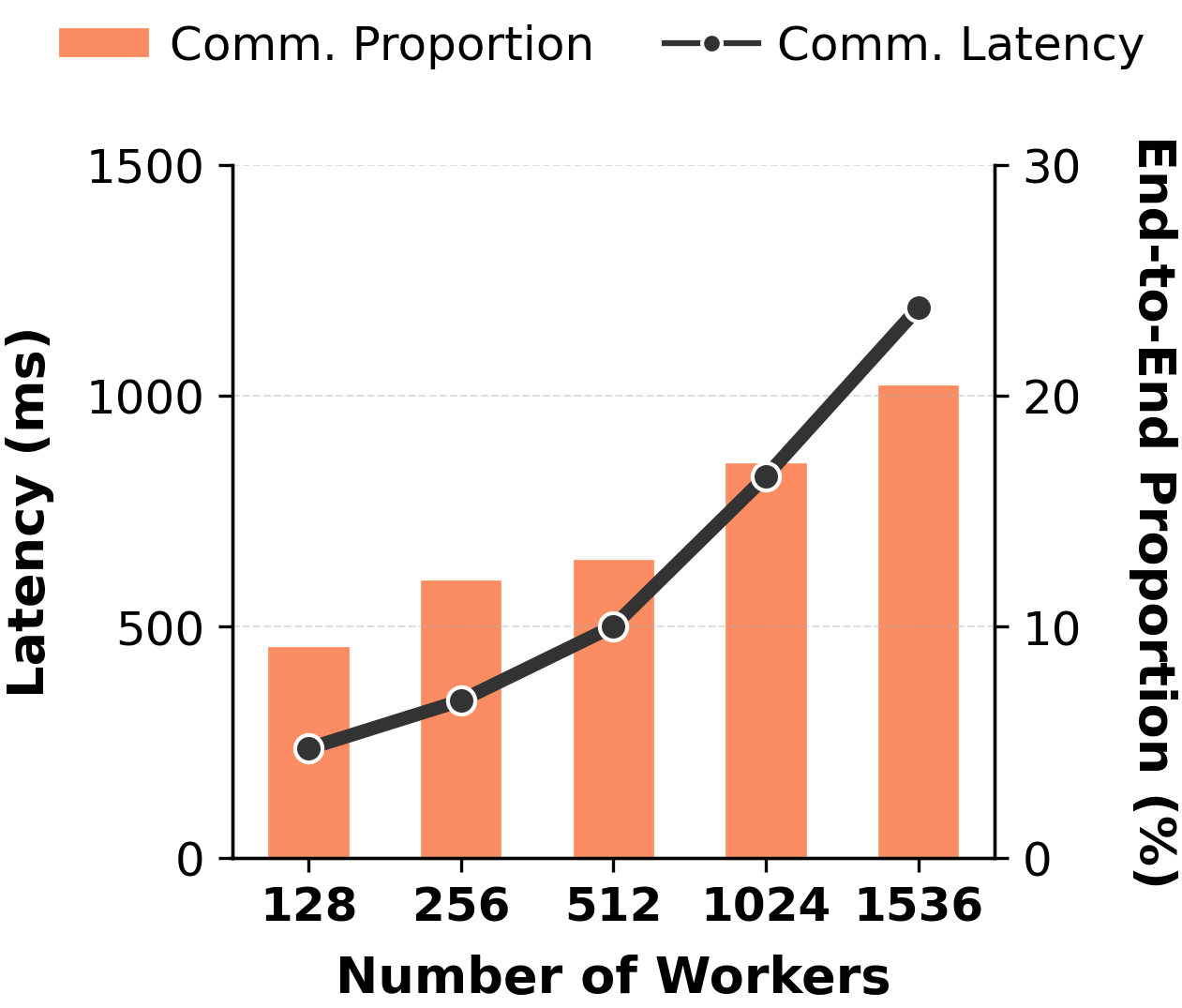}
        \caption{Communication Overhead}
        \label{fig:comm_overhead}
    \end{subfigure}
    
    \caption{Impact of cluster scale on sparse lookup and communication overhead. 
    }
    \label{fig:bottleneck}
    \vspace{-3mm}
\end{figure}

\subsection{Bottlenecks in Large-scale Embedding Training}
Although hybrid decentralized architecture effectively manages the memory footprint of colossal embedding tables, we highlight fundamental scaling barriers of decentralized embedding training shift from memory and computation to data movement.
Driven by the scaling law, lookup and communication overhead grow super-linearly with worker count.
\textbf{Lookup Bottleneck.} Embedding lookup workflow involves CPU-side data preprocessing, distributed key routing, embedding retrieval, and H2D transfers \cite{wang2024oper}. 
These operations incur negligible overhead in small-scale deployments. However, as the training cluster scales to thousands of workers, coupled with larger batch sizes and longer behavior sequences, the cumulative latency is drastically amplified. 
To verify the efficiency bottleneck under large-scale embedding training, we conduct preliminary experiments on an industrial-grade NPU cluster. As illustrated in Fig. 
\ref{fig:bottleneck}(a), lookup latency accounts for only 24.4\% of total training time when the number of workers is 128. This proportion surges to 49.6\% as the number of workers increases to 1,536. 

\textbf{Communication Bottleneck.} Due to model parallelism, the workers exchange required embedding vectors and corresponding gradients via All2All communication primitive. 
Under fully connected peer-to-peer transmission, connection complexity grows quadratically with the number of workers \cite{zhang2025two}. As shown in Fig. \ref{fig:bottleneck}(b), absolute communication latency exhibits a significant upward trend with cluster expansion. 
Moreover, the time proportion increases continuously from 9.2\% to 20.5\%, even with high-speed interconnection networks.
Under strict synchronous training paradigm, all workers must complete data exchange before entering the subsequent dense computation, which forces expensive computing resources to be idle and creates extensive pipeline bubbles \cite{guo2025adaptis}. 

\subsection{Limitations of Existing Works}

Asynchronous training has been proposed to mitigate the lookup bottlenecks \cite{lian2022persia,miao2025efficient,su2022gba,huang2021hierarchical}. 
Other pipeline parallelism schemes \cite{zhang2022picasso,adnan2024heterogeneous,kwon2022training} decouple data loading from model computation by prefetching embedding vectors for future batches while the worker computes the current batch. 
However, these solutions may cause parameter staleness and fundamentally lack reproducibility, thus compromising model convergence \cite{agarwal2023bagpipe}.
Furthermore, the prohibitive communication bottleneck also continues to limit their effectiveness in decentralized embedding training.
Even if local lookup latency is hidden, scaling to thousands of workers still exacerbates the communication overhead imposed by the All2All primitive.
Another category of works reduces communication overhead via embedding compression, such as hashing \cite{lai2023adaembed,liu2025cafe+,li2024mixed,huang2026fusedrec}, quantization \cite{feng2024accelerating,zhou2024dqrm}, and tensor-train (TT) decomposition \cite{wang2022rec,wang2024accelerating}.
They focus on reducing the absolute volume of transferred data by representing embedding tables in compact forms, but inevitably introduce approximation error. Even minor accuracy degradation (\eg, 0.1\%) is unacceptable in industrial recommendations since it can directly translate to significant revenue loss.
Some literature \cite{zha2023pre,liu2024embedding,zha2022autoshard,wang2024oper,mudigere2022software,zeng2024accelerating} has optimized embedding table sharding, placement, and scheduling to balance workloads and reduce communication hot-spots, but does not directly eliminate the lookup and communication latency exposed by synchronous execution. 
More importantly, the above methods are designed for small-scale training and cannot address the inherent challenges of decentralized embedding training with over \textit{O(1k)} workers.

The most recent work is the two-dimensional sparse parallelism \cite{zhang2025two}, which implements intra-group model parallelism and inter-group data parallelism. Although this method restricts the communication domain and peak memory footprint for large-scale embedding training, it is constrained by synchronization paradigms and leaves expensive computing resources idle during communication, resulting in low hardware utilization and limited latency reduction.
Besides, two-stage gradient aggregation via modified network topology will lead to model convergence deviation.
As summarized in Table \ref{tab:comparison}, most existing works have succumbed to the accuracy-throughput dilemma. 
They either sacrifice strict consistency for higher training throughput or only address a single bottleneck (lookup or communication) while leaving the other unoptimized. 
Besides, these works are inherently confined to small-scale training configurations and thus suffer from severe scalability degradation as the cluster scale expands.

\section{Overview of NestPipe}\label{sec:design}

Existing solutions always focus on reducing the absolute magnitude of lookup or communication overhead. 
For large-scale decentralized embedding training, the key problem is not only how much sparse overhead exists in total, but how much of it remains exposed on the end-to-end workflow.
From this perspective, we design a decentralized framework, called NestPipe, which addresses both lookup and communication bottlenecks of large-scale embedding training while maintaining synchronous semantics.
As illustrated in Fig. \ref{fig:overview}, our nested pipelining design exploits two sparse parallelism opportunities at different spatial granularities.


At the inter-batch level, DBP leverages distinct hardware resources and decomposes the embedding lookup workflow into five-stage pipeline across consecutive batches (Section \ref{sec:DBP}). At the intra-batch level, FWP exploits the parameter freezing phenomenon in micro-batch training to decouple All2All communication from dense computation (Section \ref{sec:FWP}). 
While batch 4 undergoes forward and backward computation, batch 1/2/3 concurrently progress through key routing, data H2D, and data prefetch stages, which fully hides lookup latency. 
Meanwhile, entering the fine-grained pipeline within a single batch (\eg, batch 4 with four micro-batches), FWP overlaps the embedding All2All communication of each micro-batch with the dense computation of its adjacent micro-batches via carefully orchestrated scheduling streams.
In other words, DBP addresses the latency before sparse embeddings become ready in HBM, whereas FWP addresses the latency after sparse embeddings need to be exchanged across workers.


Such a hierarchical sparse parallelism design is how NestPipe fulfills three design goals of \emph{efficiency}, \emph{consistency}, and \emph{scalability}, as empirically verified in Section \ref{sec:Evaluation}.
The combination of DBP and FWP simultaneously addresses lookup and communication bottlenecks exposed by decentralized embedding training. 
Moreover, DBP eliminates the embedding staleness in naive pipelining by dual-buffer synchronization before each batch initiates its forward propagation, while FWP guarantees that communication overlap occurs only within a semantically valid frozen window, thus maintaining strict consistency to synchronous training.
Finally, unlike prior solutions limited to small-scale settings, NestPipe preserves near-linear scaling efficiency even with up to thousands of workers.



\begin{figure}
    \centering
    \includegraphics[width=0.47\textwidth]{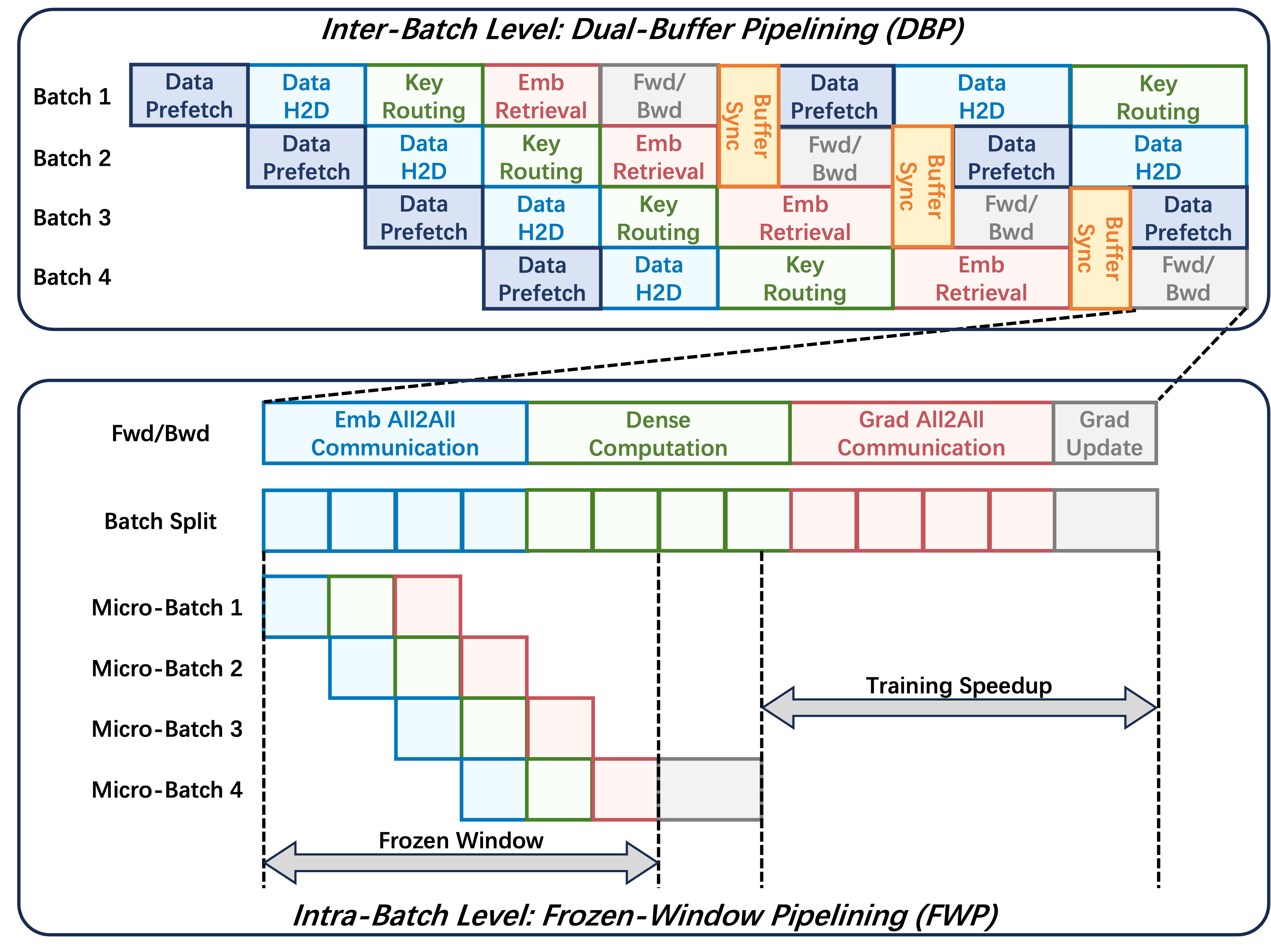} 
    \caption{Overview of NestPipe. 
    }
    \label{fig:overview}
    \vspace{-3mm}
\end{figure}

\section{Inter-Batch Level: Dual-Buffer Pipelining}\label{sec:DBP}
\subsection{Staleness-free Five-stage Parallelism}
To tackle the issue of lookup overhead, we design Dual-Buffer Pipelining (DBP) strategy.
Sparse lookup is not a single operation but a multi-stage data-movement. Before a batch can enter effective model computation, we must complete CPU-side data preprocessing, distributed key routing,
embedding retrieval, and H2D transfers.
We analyze the dependency chains in decentralized training architecture and identify distinct hardware resources (\eg, CPU for preprocessing, network for communication, HBM for embedding storage, accelerators for computation). This allows us to transform the serial workflow into parallel stages and leverage the predictable nature of future batches to optimize data movement.
However, the naive pipeline design will introduce the risk of embedding staleness. 
Since embedding accesses follow a highly skewed distribution \cite{miao2025efficient}, a small subset of popular embeddings frequently participates in the recommendation training of consecutive batches.
Let $\mathcal{B}_t$ denote the batch at training step $t$.
The prefetched embeddings for $\mathcal{B}_{t}$ might become stale if they are not updated by the batch $\mathcal{B}_{t-1}$ still in the pipeline \cite{jung2026mitigating}. To this end, DBP integrates dual-buffer synchronization and enables a staleness-free five-stage pipeline. We detail each stage and their coordination below:


\begin{itemize}
    \item \textbf{Data Prefetch}: The workers read raw user-item interaction logs and prefetch the structured data $\mathcal{B}_t$ into \textit{pinned memory} (non-pageable memory), which eliminates OS-level memory paging overhead for fast subsequent data transfers \cite{agarwal2014design}. 
    \item \textbf{Data H2D}: The prepared data $\mathcal{B}_t$ is asynchronously copied to the HBM. Benefiting from pinned memory and direct memory access \cite{thieu2025dcma}, this stage significantly reduces H2D latency compared to transfers from standard host memory.
    \item \textbf{Key Routing}: The sparse keys within the batch $\mathcal{B}_t$ are first deduplicated to reduce redundant communication and then partitioned into buckets based on embedding table sharding rules, which is aligned with model parallelism. The bucketed keys are routed to destination workers hosting the corresponding embedding vectors via All2All communication. Since keys are orders of magnitude smaller than embedding vectors or gradients, this key transmission remains lightweight and rarely bottlenecks the overall system performance.
    \item \textbf{Embedding Retrieval}: Upon receiving the keys, each destination worker again performs sparse key deduplication to eliminate redundant keys from different peer source workers, thereby directly reducing the overhead of subsequent embedding lookups. The destination workers then query the embedding vectors indexed by these deduplicated keys from the embedding tables. The retrieved embeddings are transferred from host memory (DRAM) to device memory (HBM). The destination workers perform dual-buffer synchronization for two consecutive training batches $\mathcal{B}_t$ and $\mathcal{B}_{t-1}$ to propagate the latest parameter updates (detailed in Section \ref{subsec:dual-buffer}).
    \item \textbf{Fwd/Bwd}: 
    The All2All communication primitive is adopted to send the synchronized embeddings back to the source workers that initially requested them. 
    For now, each source worker obtains the complete set of embedding vectors required by its local batch $\mathcal{B}_t$. 
    They perform forward and backward propagation to compute gradients for both embeddings and dense layers, which are subsequently synchronized among all workers via All2All and AllReduce, respectively. Finally, the updated embedding vectors are written back to host memory.
\end{itemize}
Each stage is designed to utilize distinct hardware resources to avoid resource contention. 
The stage division follows the principle of fine-grained resource decoupling and overlapping execution. 
Fewer stages would limit overlapping opportunities, while additional stages may introduce unnecessary complexity without proportional performance gains \cite{zhang2025td}.
By parallelizing these stages across multiple training batches, NestPipe strategically hides the lookup latency of data preprocessing, key routing, and embedding retrieval. The concurrent execution turns idle waiting time into productive operation \cite{guo2025adaptis}, which is particularly effective in large-scale embedding access scenarios where these latencies are more pronounced. 

\subsection{Dual-buffer Synchronization}\label{subsec:dual-buffer}
To avoid embedding staleness, DBP strategy maintains two HBM buffers under a producer-consumer pattern:
\begin{itemize}
    \item \textbf{Active HBM Buffer} serves the forward and backward propagation for the current batch. The gradients are applied directly to the embedding vectors in this buffer to keep them up-to-date.
    \item \textbf{Prefetch HBM Buffer} is used to preload embeddings for the next batch while the current batch performs forward and backward propagation.
\end{itemize}
The workflow of dual-buffer synchronization is shown in Fig. \ref{fig:dual-buffer}.
After retrieving the embedding vectors in active buffer, batch $\mathcal{B}_{t-1}$
performs forward and backward propagation. In parallel, the prefetch buffer preloads the embedding vectors required by batch $\mathcal{B}_{t}$ to reduce the lookup latency. 
Before batch $\mathcal{B}_{t}$ in prefetch buffer begins its forward propagation, DBP strategy takes the intersection of the two buffers ($\mathcal{B}_{t-1} \cap \mathcal{B}_{t}$). 
In practice, a dedicated kernel computes the intersection of compact key sets between consecutive batches, rather than full embedding vectors. 
Then, DBP strategy performs embedding synchronization of active buffer and prefetch buffer via fast device-to-device memory copies. The overhead of dual-buffer synchronization is typically lower than 2ms, which can be fully overlapped with other concurrent stages.

After the buffer synchronization completes, the prefetch buffer contains the latest parameters updated by batch $\mathcal{B}_{t-1}$ in active buffer.
Then, batch $\mathcal{B}_{t}$ can start its All2All communication immediately with validated fresh embeddings. Meanwhile, the updated embeddings from the finished batch $\mathcal{B}_{t-1}$ are written back to the host memory to maintain the consistency of the complete embedding table in the hierarchical storage architecture.
DBP strategy then switches the active buffer pointer. In other words, the roles of the active buffer and prefetch buffer alternate for each subsequent batch. The original active buffer serving batch $\mathcal{B}_{t-1}$ is converted to the prefetch buffer for the next batch ($\mathcal{B}_{t+1}$), while the synchronized prefetch buffer for batch $\mathcal{B}_{t}$ takes over as the new active buffer.
The embedding synchronization at buffer intersection guarantees parameter freshness without breaking pipeline parallelism.


\begin{figure}
    \centering
    \includegraphics[width=0.43\textwidth]{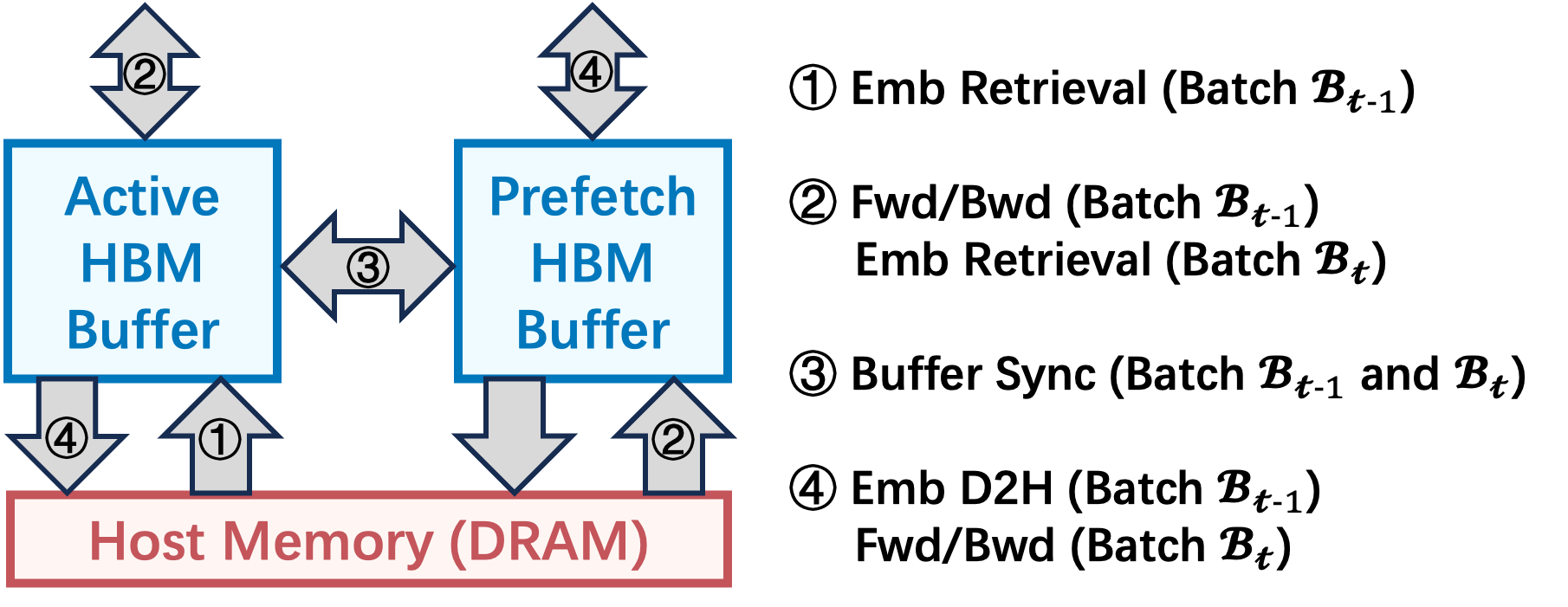} 
    \caption{Dual-buffer synchronization in DBP strategy.
    }
    \label{fig:dual-buffer}
    \vspace{-2mm}
\end{figure}

\begin{figure*}
    \centering
    \includegraphics[width=\textwidth]{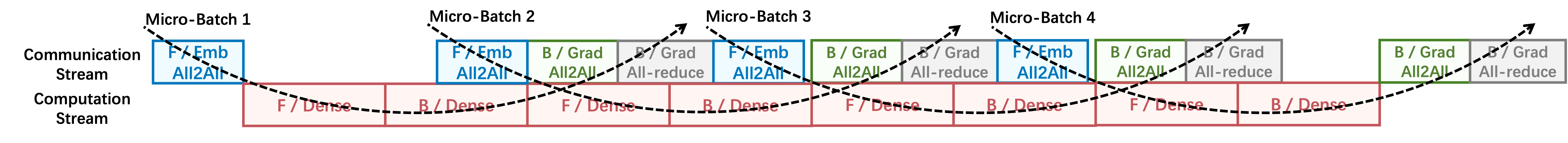} 
    \caption{Implementation of FWP strategy through coordinated communication and computation stream scheduling.
    }
    \label{fig:stream}
    \vspace{-3mm}
\end{figure*}

\section{Intra-Batch Level: Frozen-Window Pipelining}\label{sec:FWP}
\label{FWP}

\subsection{Intra-batch Communication Overlap}
Although the dual-buffer pipelining design effectively masks the lookup latency, the All2All communication during forward and backward propagation still hinders efficient decentralized training.
A straightforward intuition is to extend the five-stage pipeline into a six-stage version by parallelizing the dense computation of the batch $\mathcal{B}_{t-1}$ and the All2All communication of the next batch $\mathcal{B}_{t}$. 
In the original DBP strategy, embedding All2All communication of $\mathcal{B}_{t}$ is only initiated after the dense computation of $\mathcal{B}_{t-1}$ is completed, ensuring embedding freshness.
If DBP is extended to the six-stage pipeline, the gradient generated by $\mathcal{B}_{t-1}$ cannot leverage dual-buffer synchronization to update the embedding vectors that have already been transmitted via All2All communication. As a result, the embeddings used in the forward propagation of $\mathcal{B}_{t}$ do not reflect the most recent weight updates.
Inconsistent gradients accumulate across iterations, which leads to undesirable accuracy degradation and even impairs the convergence stability, especially for generative recommendation models that are sensitive to parameter consistency.
The six-stage extension appears to maximize the overlap of computation and communication, but it results in the one-step asynchrony issue in practice, invalidating the efficiency gain from parallelism.


In fact, the one-step asynchrony issue arises because we always attempt to keep the embeddings up-to-date while simultaneously transmitting them to remote workers, which is quite challenging during sparse pipelining.
In particular, we revisit the updating frequency of sparse parameters and provide a fundamentally different perspective: if the parameters do not change during a specific window, then the parameters transmitted during that window are naturally ``fresh'' and there is no newer version to miss.
On this basis, we note the \textit{parameter freezing phenomenon} in the micro-batch training, where the forward and backward propagation of a single micro-batch calculates the embedding gradients without performing the actual parameter updating \cite{wang2024accelerating,wan2026revisiting}, which creates an implicit window for parallelizing communication and computation without introducing training inconsistency.

Based on the above insight, NestPipe introduces Frozen-Window Pipelining (FWP) to address the communication bottleneck via fine-grained optimization. By splitting the training batch $\mathcal{B}_t$ into $N$ micro-batches $\{\mathcal{M}_1^t, \mathcal{M}_2^t, \cdots, \mathcal{M}_N^t\}$, we can decouple the All2All communication from the dense computation. 
As illustrated in Fig. \ref{fig:overview}, we perform the embedding All2All communication, dense computation, and gradient All2All communication of all micro-batches in parallel. 
The embeddings will not be updated during the frozen window.
The gradients are applied only after all micro-batches of batch $\mathcal{B}_t$ complete their gradient computation. 
As a result, the embedding vectors used by any micro-batch $\mathcal{M}_i^t$ ($1\le i\le N$) are always the latest version, which avoids the one-step asynchrony and maintains equivalence to synchronous training. 

\subsection{Stream Scheduling}
At the implementation level, NestPipe optimizes resource orchestration through stream scheduling.
As shown in Fig. \ref{fig:stream}, FWP strategy is implemented by two independent and coordinated execution streams: 
\begin{itemize}
    \item \textbf{Computation Stream} is dedicated to the dense forward and backward propagation of micro-batches. The gradient calculation and parameter update operations are scheduled on the computing cores to maximize resource utilization. 
    \item \textbf{Communication Stream} is responsible for all data transmission operations in the training process, including the All2All communication of embedding parameters/gradients and All-Reduce communication of dense layers, which is scheduled on the interconnect network.
\end{itemize}
The scheduling principle is that communication should be launched as early as possible within the frozen window, and computation should consume ready micro-batches without waiting for unrelated communication to finish.
Once the sparse embeddings of a micro-batch $\mathcal{M}_i^t$ become available in local HBM, the computation stream proceeds to dense computation independently while the communication stream advances the All2All communication for subsequent micro-batches $\mathcal{M}_{i+1}^t$ \cite{zhang2025comet}.
These streams coordinate through carefully designed synchronization points that align with the micro-batch boundaries. 
With the help of decoupled stream scheduling, FWP strategy eliminates hardware contention between computation and communication, further improving the overall hardware utilization of the cluster.

\subsection{Sample Clustering}
In our FWP strategy, the micro-batch size has a critical impact on the trade-off between physical communication overhead and exposed ratio.
Let $N$ denote the number of micro-batches within a batch ($N=4$ in Fig. \ref{fig:stream}). Each training step involves $2N$ All2All communications, including $N$ embedding All2All and $N$ gradient All2All communications. Since FWP enables full overlap of all intra-batch communication except for the first and last boundary communication, the theoretical exposed communication ratio is exactly $1/N$. 
Although a smaller micro batch size (larger $N$) reduces the exposed boundary communication ratio, key deduplication for embedding communication is inefficient.
Since sparse keys are deduplicated within individual small micro-batches, redundant keys scattered across different micro-batches cannot be eliminated in naive micro-batch splitting, leading to repeated transmission of the same embeddings in $2N$ All2All operations.
As a result, the inflated communication payload may exceed the available computation window, causing overlap to collapse and ultimately breaking the intended benefit of FWP strategy.

To strengthen deduplication efficiency, we incorporate a lightweight key-centric sample clustering scheme into FWP. Specifically, we group samples that share more sparse keys into the same micro-batch, maximizing key redundancy within micro-batch \cite{wang2024accelerating}.
This reduces repeated embedding transmission and helps FWP achieve its theoretical exposed ratio in practice.
More importantly, sample clustering only changes the order of embedding and gradient communication, without modifying the embedding values used in forward propagation or the final gradients for each key.
Therefore, it does not affect the model convergence behavior, which is also empirically validated in Section \ref{sec:Evaluation}.
To ensure clustering operation does not increase the end-to-end training time, the sample clustering can be executed asynchronously on the CPU as part of the data preprocessing stage in DBP strategy, or pre-computed offline. By decoupling it from the active computation stream, the corresponding overhead is successfully hidden behind the concurrent dual-buffer pipeline. 


\section{Theoretical Consistency Analysis}\label{sec:convergence}

In this section, we formally prove that NestPipe maintains mathematical equivalence to standard synchronous training. 

\begin{definition}(Synchronous Training Consistency)
$W_t = (\theta_t, E_t)$ denotes the model parameters at step $t$, where $\theta_t$ represents the dense layers and $E_t = \{e_k^t\}_{k \in \mathcal{V}}$ represents the sparse embedding table over vocabulary $\mathcal{V}$. Let $\mathcal{K}(\mathcal{B}_{t}) \subseteq \mathcal{V}$ define the set of distinct sparse keys accessed by batch $\mathcal{B}_{t}$. Given a learning rate $\eta$ and a loss function $F$, the standard synchronous update is:
\begin{equation}
W_{t+1} = W_t - \eta \frac{1}{|\mathcal{B}_{t}|} \sum_{\xi \in \mathcal{B}_{t}} \nabla F(W_t, \xi)
\label{eq:sync}
\end{equation}
At each step $t$, any deviation from this formulation breaks training consistency, such as computing gradients using stale weights $W_{t-\tau}$ ($\tau > 0$) or altering gradient aggregation logic.
\end{definition}

\begin{proposition}
(Consistency of DBP)
Under the DBP strategy with dual-buffer synchronization, the parameters available to batch $\mathcal{B}_{t+1}$ are exactly $W_{t+1}$.
\end{proposition}
\begin{proof}
    The prefetch buffer $H_{\mathrm{pref}}$ loads embeddings for $\mathcal{B}_{t+1}$ while the active buffer $H_{\mathrm{act}}$ serves the forward and backward propagation of $\mathcal{B}_{t}$. 
For each $k \in \mathcal{K}(\mathcal{B}_{t+1})$, we have:
\begin{equation}
e_k^{t+1} =
\begin{cases}
e_k^t, & k \notin \mathcal{K}(\mathcal{B}_{t}) \\
e_k^t - \eta \frac{1}{|\mathcal{B}_{t}|} \sum_{\xi \in \mathcal{B}_{t}} \nabla_{e_k} F(W_t, \xi), & k \in \mathcal{K}(\mathcal{B}_{t})
\end{cases}
\label{eq:dbp}
\end{equation}
For non-overlapping keys, $\nabla_{e_k} F(W_t, \xi) = \mathbf{0}$. The prefetched embedding $e_k^t$ is already up-to-date.
For overlapping keys $k \in \mathcal{K}(\mathcal{B}_{t}) \cap \mathcal{K}(\mathcal{B}_{t+1})$, the embedding $e_k^t$ in $H_{\mathrm{pref}}$ is strictly overwritten by the updated value from $H_{\mathrm{act}}$ via dual-buffer synchronization.
Combining with $\theta_{t+1}$ obtained via AllReduce synchronization, $\mathcal{B}_{t+1}$ uses the rigorously validated $W_{t+1} = (\theta_{t+1}, E_{t+1})$ rather than $W_{t-\tau}$ to compute gradients.
\end{proof}

\begin{proposition}
(Consistency of FWP)
Under the FWP strategy with micro-batch training and sample clustering, the parameter update is exactly equivalent to Eq.~\eqref{eq:sync}.
\end{proposition}
\begin{proof}
Consider the batch $\mathcal{B}_t$ with $N$ micro-batches $\{\mathcal{M}_1^t, \mathcal{M}_2^t, \cdots, \mathcal{M}_N^t\}$, such that $\mathcal{B}_t = \bigcup_{i=1}^N \mathcal{M}_i^t$. 
Key-centric sample clustering produces an alternative partition $\{\tilde{\mathcal{M}}_1^t, \tilde{\mathcal{M}}_2^t, \cdots, \tilde{\mathcal{M}}_N^t\}$.
After the $N$-th micro-batch transmits its gradients, FWP updates model parameters as follows:
\begin{align}
    W_{t+1} &= W_t - \eta \frac{1}{|\mathcal{B}_t|} \sum_{i=1}^{N} \sum_{\xi \in \tilde{\mathcal{M}}_i^t}\nabla F(W_t,\xi) \\
    &\overset{(a)}{=} W_t - \eta \frac{1}{|\mathcal{B}_t|} \sum_{i=1}^{N} \sum_{\xi \in \mathcal{M}_{i}^t}\nabla F(W_t,\xi)  \\
    &\overset{(b)}{=} W_t - \eta \frac{1}{|\mathcal{B}_t|} \sum_{\xi \in \mathcal{B}_t}\nabla F(W_t,\xi)
\end{align}
where $(a)$ follows from summation over disjoint partitions of the same set $\mathcal{B}_t$. Altering the sample order through clustering does not change the final sum of original gradients.
$(b)$ follows $\mathcal{B}_t = \bigcup_{i=1}^N \mathcal{M}_i^t$ as no parameter update occurs between micro-batches. Thus, $\nabla F(W_t, \xi)$ for $\xi \in \mathcal{M}_{i}^t$ is computed on the identical frozen $W_t$, regardless of the micro-batch index $i$.
\end{proof}

\begin{corollary}
(Consistency of NestPipe)
The nested combination of DBP and FWP satisfies Definition 1 at each step $t$.
\end{corollary}
\begin{proof}
Proposition 1 guarantees that DBP establishes the correct initial state at inter-batch boundaries.
At each step $t$, the latest $W_t$ is established before $\mathcal{B}_{t}$ begins forward propagation. 
Proposition 2 guarantees that FWP preserves gradient equivalence within batch $\mathcal{B}_{t}$.
The gradients computed during the frozen window equals the full-batch gradients. 
Thus, their nested composition satisfies Eq.~\eqref{eq:sync} at each training step $t$.
\end{proof}

\section{Performance Evaluation}\label{sec:Evaluation}
\begin{table}[t]
\centering
\caption{Overall latency (ms) of different methods on GPU and NPU clusters, and ablation study of DBP and FWP. 
}
\label{tab:overall_performance_1536}
\resizebox{0.5\textwidth}{!}{
\begin{tabular}{@{}ccc c c c ccc c@{}}
\toprule
\multirow{2}{*}{\textbf{Cluster}} & \multirow{2}{*}{\textbf{Dataset}}  & \multirow{2}{*}{\textbf{Method}} & \textbf{Step} & \multirow{2}{*}{\textbf{Speedup}} & \multicolumn{2}{c}{\textbf{Ablation Study}}  \\ \cmidrule(lr){6-7}
& &  & \textbf{Latency} & & \textbf{Lookup} & \textbf{Comm.}  \\ \midrule

 & \multirow{4}{*}{Industrial} 
& TorchRec       & 5793.83 & 1.00$\times$  & 2870.99 & 1207.85  \\
1536 &  & 2D-SP          & 4914.01 & 1.18$\times$  & 2766.68 & 438.36  \\
 NPUs &  & UniEmb       & 2919.76 & 1.98$\times$  & 36.21   & 1169.01  \\
&  & \cellcolor{gray!20}\textbf{NestPipe}  &\cellcolor{gray!20}\textbf{1895.98} & \cellcolor{gray!20}\textbf{3.06$\times$}    & \cellcolor{gray!20}\textbf{30.19}   & \cellcolor{gray!20}\textbf{154.23}  \\ \midrule

 &  
& TorchRec       & 4207.94 & 1.00$\times$ & 856.63 & 312.19 \\
128 & KuaiRand & 2D-SP          &4052.98 & 1.04$\times$  & 907.87 & 107.23   \\
GPUs & -27K & UniEmb       & 3384.05 & 1.24$\times$  &  23.49  & 327.29 \\
&  & \cellcolor{gray!20}\textbf{NestPipe}  &\cellcolor{gray!20}\textbf{3090.45} & \cellcolor{gray!20}\textbf{1.36$\times$}    & \cellcolor{gray!20}\textbf{12.77}   & \cellcolor{gray!20}\textbf{38.95}  \\

\bottomrule
\end{tabular}}
\end{table}

Our evaluation aims to answer five research questions.
\begin{itemize}
    \item \textbf{RQ1 (Efficiency):} Does NestPipe improve end-to-end recommendation training efficiency?
    \item \textbf{RQ2 (Consistency):} Does NestPipe preserve training semantics relative to standard synchronous baseline?
    \item \textbf{RQ3 (Scalability):} Does NestPipe maintain high scaling factor as the training cluster scales expand?
    \item \textbf{RQ4 (Sensitivity):} How does NestPipe perform across different micro-batch sizes and workload characteristics?
    \item \textbf{RQ5 (Orthogonality):} Is NestPipe complementary to existing communication-reduction techniques?
\end{itemize}

\begin{figure*}[htbp]
    \centering
    \includegraphics[width=\textwidth]{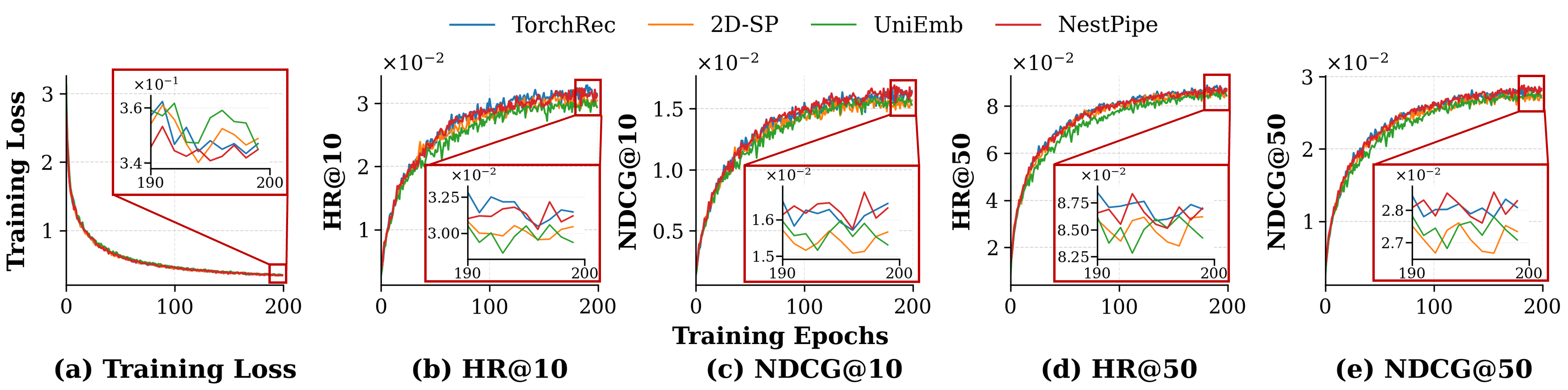}
    \caption{
    The training loss and accuracy curve of different methods.
    }
    \label{fig:accuracy}
\end{figure*}

\begin{table*}[t]
\centering
\caption{Scaling performance comparison of different methods. QPS is reported in the scale of $10^5$. The speedup is computed against TorchRec under the same cluster configuration. The scaling factor is normalized to the 128-worker baseline.
}
\label{tab:scaling_performance}
\resizebox{0.9\textwidth}{!}{
\begin{tabular}{@{} c ccc ccc ccc ccc @{}}
\toprule
\multirow{2}{*}{\textbf{\# Worker}} & \multicolumn{3}{c}{\textbf{TorchRec}} & \multicolumn{3}{c}{\textbf{2D-SP}} & \multicolumn{3}{c}{\textbf{UniEmb}} & \multicolumn{3}{c}{\textbf{NestPipe}} \\ 
\cmidrule(lr){2-4} \cmidrule(lr){5-7} \cmidrule(lr){8-10} \cmidrule(l){11-13}
& \textbf{QPS} & \textbf{Speedup}  & \textbf{Scaling}  & \textbf{QPS} & \textbf{Speedup} & \textbf{Scaling} & \textbf{QPS} & \textbf{Speedup}  & \textbf{Scaling} & \textbf{QPS} & \textbf{Speedup}  & \textbf{Scaling} \\ \midrule

128  & 0.26 &1.00$\times$  &    -    & 0.27 &1.04$\times$ &   -     & 0.33 &1.27$\times$ &    -    & \cellcolor{gray!20}\textbf{0.37} &\cellcolor{gray!20}\textbf{1.42$\times$} & - \\
256  & 0.47 &1.00$\times$ & 91.23\% & 0.49 &1.04$\times$ & 91.36\% & 0.63 &1.34$\times$ & 96.13\% & \cellcolor{gray!20}\textbf{0.72} &\cellcolor{gray!20}\textbf{1.53$\times$} & \cellcolor{gray!20}\textbf{98.39\%} \\
512  & 0.68 &1.00$\times$ & 66.64\% & 0.75 &1.10$\times$ & 69.64\% & 1.17 &1.72$\times$ & 89.21\% & \cellcolor{gray!20}\textbf{1.43} &\cellcolor{gray!20}\textbf{2.10$\times$} & \cellcolor{gray!20}\textbf{97.33\%} \\
1024 & 1.09 &1.00$\times$ & 53.36\% & 1.22 &1.12$\times$ & 56.67\% & 2.04 &1.87$\times$ & 78.00\% & \cellcolor{gray!20}\textbf{2.80} &\cellcolor{gray!20}\textbf{2.57$\times$} & \cellcolor{gray!20}\textbf{95.63\%} \\
1536 & 1.36 &1.00$\times$ & 44.34\% & 1.60 &1.18$\times$ & 49.32\% & 2.65 &1.98$\times$ & 67.62\% & \cellcolor{gray!20}\textbf{4.14} &\cellcolor{gray!20}\textbf{3.06$\times$} & \cellcolor{gray!20}\textbf{94.07\%} \\ 

\bottomrule
\end{tabular}}
\end{table*}

\subsection{Experimental Setup}
\label{sec:exp_setup}
\textbf{Training Specification.}
We evaluate our proposed framework on two industrial-grade production clusters, \ie, 1536-NPU cluster and 128-GPU cluster.
All experiments are conducted on \textit{KuaiRand-27K} \cite{gao2022kuairand} and \textit{Industrial} datasets.
Considering that publicly available datasets are insufficient in cardinality to meet \textit{O(1k)}-scale training requirements, we include an industrial recommendation dataset to reflect large-scale data distributions and sparsity patterns.
We adopt HSTU~\cite{zhai2024actions} and FUXI \cite{ye2025fuxi} as the backbone models, which are widely used in industrial recommendation tasks. 
Unless otherwise specified, all experiments are conducted on the NPU cluster with 1,536 workers, using HSTU on the Industrial dataset.

\textbf{Baselines.}
To rigorously assess the training performance, we compare NestPipe against the following three baselines. 
\begin{itemize}[leftmargin=*]
    \item \textbf{TorchRec}~\cite{10.1145/3523227.3547387} as the de facto open-source framework provides full-stack sparsity primitives needed for large-scale embedding tables and utilizes the hybrid decentralized architecture to accelerate embedding training.
    \item \textbf{2D-SP}~\cite{zhang2025two} is the SOTA two-dimensional sparse parallelism solution, which mitigates communication overhead by restricting All2All communication domain within local worker groups. Following the reported optimal configuration, the number of parallelism groups is set to 4.
    \item \textbf{UniEmb} is the industrial-grade distributed training engine that integrates mature sparse optimization commonly adopted in production recommendation systems, including embedding sharding, dynamic hash table, operator fusion optimization, and asynchronous prefetch pipeline.
\end{itemize}

\textbf{Evaluation Metrics.}
\textit{(1) Step latency} records the average end-to-end latency of repeated training steps, which mainly includes computing, lookup, and exposed communication time.
\textit{(2) HR@$K$ and NDCG@$K$} are adopted to evaluate model accuracy, which respectively measure the hit rate of relevant items and ranking quality within the top-$K$ list.
\textit{(3) Throughput (QPS)} is defined as the number of processed samples per second. \textit{(4) Resource utilization ratio} is defined as the percentage of time that computing cores remain active within a given training period. \textit{(5) Exposed comm. ratio} is the percentage of physical All2All communication latency that is not hidden behind dense computation. 

\subsection{RQ1: End-to-end Efficiency}
We first evaluate the end-to-end training efficiency on both NPU and GPU clusters. Table~\ref{tab:overall_performance_1536} summarizes the step latency and the major sparse overheads under two representative settings, including HSTU on the Industrial dataset and FUXI on the KuaiRand-27K dataset.
We observe that NestPipe achieves the best end-to-end performance compared to baselines, delivering 3.06$\times$ speedup on NPU cluster and 1.36$\times$ speedup on GPU cluster.
In particular, TorchRec encounters severe lookup (2,870.99ms) and communication (1,207.85ms) bottlenecks during training with 1,536 workers.
2D-SP only mitigates communication time to 438.36ms by restricting the communication domain, while UniEmb hides lookup latency of 2,834.78ms, leaving the other bottleneck unaddressed.
The training speedup is limited to 1.18$\times$ and 1.98$\times$, respectively.
In contrast, the performance gain of NestPipe stems from the hierarchical
sparse parallelism. 
Table~\ref{tab:overall_performance_1536} also serves as an ablation study evaluating the individual contributions of DBP and FWP strategies.
DBP-only alleviates about 98\% lookup latency for different datasets by constructing a staleness-free five-stage pipeline.
FWP-only reduces the exposed comm. ratio to 13\% by overlapping All2All communication with dense computation, while the baselines (including SOTA 2D-SP method) expose 100\% of communication time.
These results confirm that NestPipe simultaneously addresses lookup and communication bottlenecks, achieving efficient embedding training via hardware-agnostic optimization.

\begin{figure}
    \centering
    \includegraphics[width=0.45\textwidth]{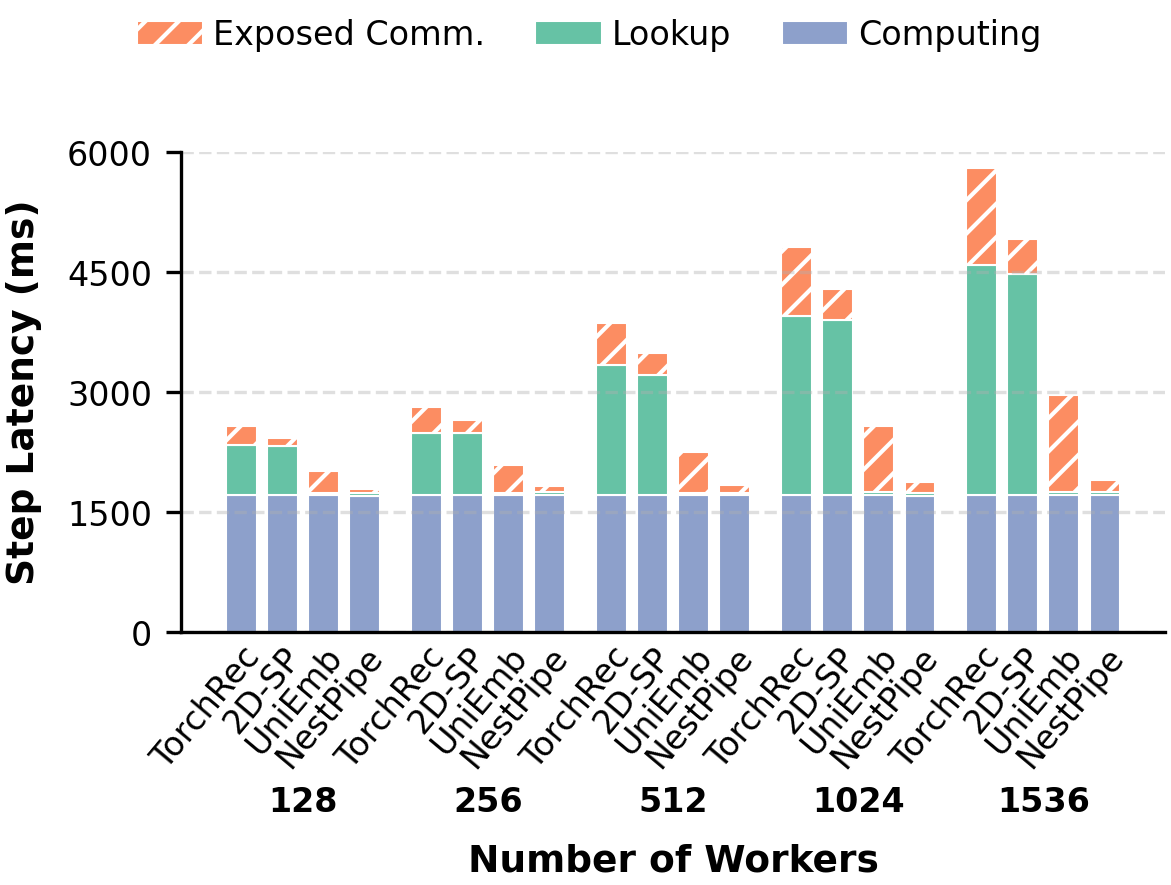} 
    \caption{Step latency breakdown of different methods for varying cluster sizes. 
    }
    \label{fig:scaling}
    \vspace{-2mm}
\end{figure}

\subsection{RQ2: Training Consistency}
To empirically verify consistency, we train the FUXI model using the KuaiRand-27K dataset. The training loss and ranking metrics of different methods are reported in Fig. ~\ref{fig:accuracy}.
NestPipe closely follows the synchronous baseline (TorchRec) throughout the training process, whereas the ranking metrics of other methods exhibit a significant decline. 
For example, the HR@10 and HR@50 metrics of UniEmb decrease by 2.1$\times10^{-3}$ and 2.7$\times10^{-3}$ respectively, indicating that even limited embedding staleness can affect recommendation quality.
Similarly, 2D-SP also experiences measurable accuracy degradation, dropping by 1.0$\times10^{-3}$ in HR@10 and 0.7$\times10^{-3}$ in NDCG@10.
Although 2D-SP restricts the All2All communication domain, it changes the original communication topology and gradient aggregation logic, which can perturb training convergence.
In contrast, the differences of NestPipe across all four ranking metrics are uniformly less than 0.3$\times10^{-3}$. 
These observations are consistent with the theoretical analysis in Section~\ref{sec:convergence}. DBP performs dual-buffer synchronization before forward propagation, and FWP implements overlap only within the frozen window where no parameter update occurs. Therefore, NestPipe preserves synchronous training semantics in practice while providing throughput gains.

\begin{figure}
    \centering
    \includegraphics[width=0.45\textwidth]{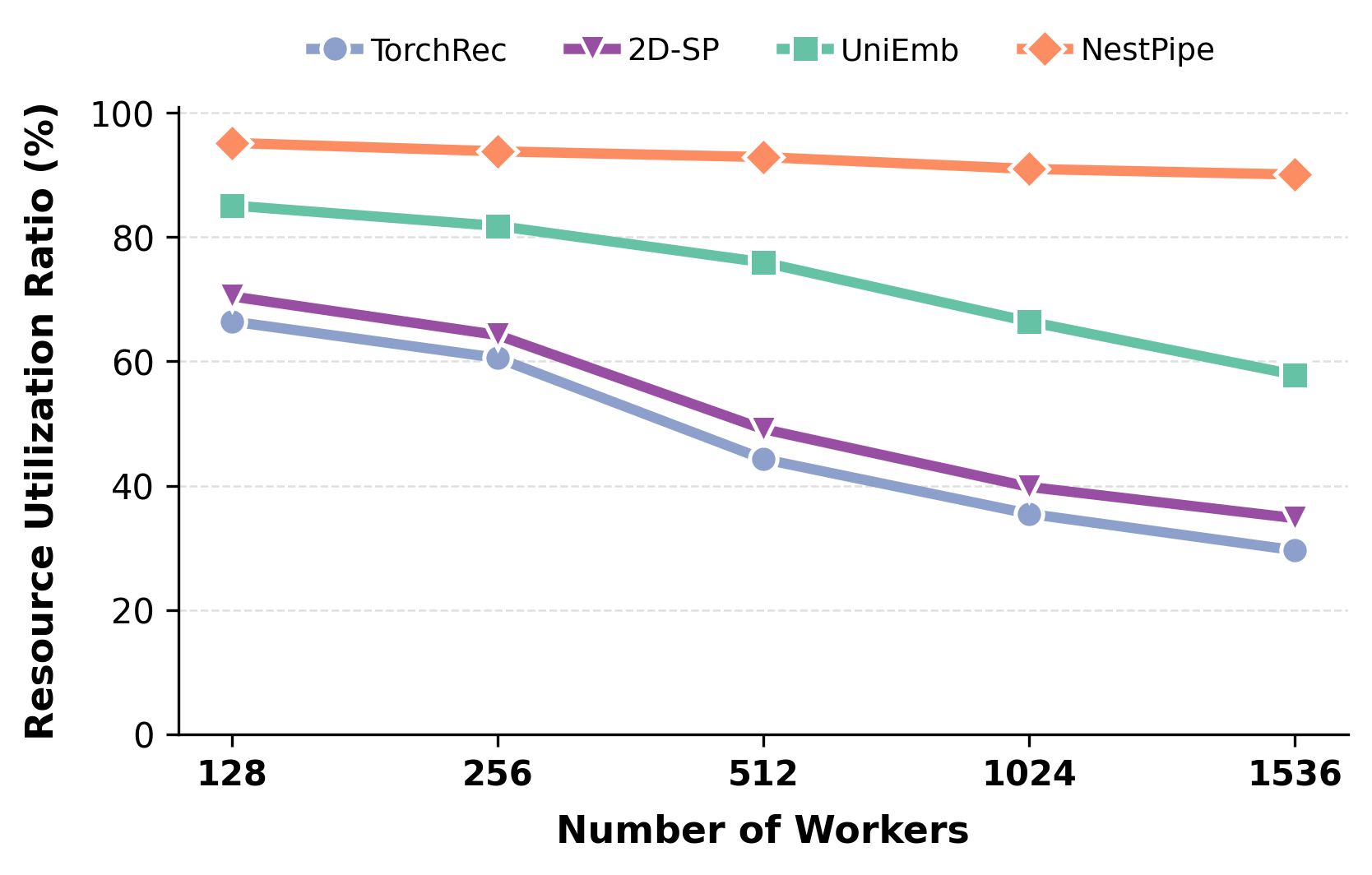} 
    \caption{Resource utilization ratio of different methods for varying cluster sizes.
    }
    \label{fig:utilization}
    \vspace{-2mm}
\end{figure}

\begin{figure}[t]
    \centering
    \begin{subfigure}{0.49\columnwidth}
        \centering
        \includegraphics[width=\linewidth]{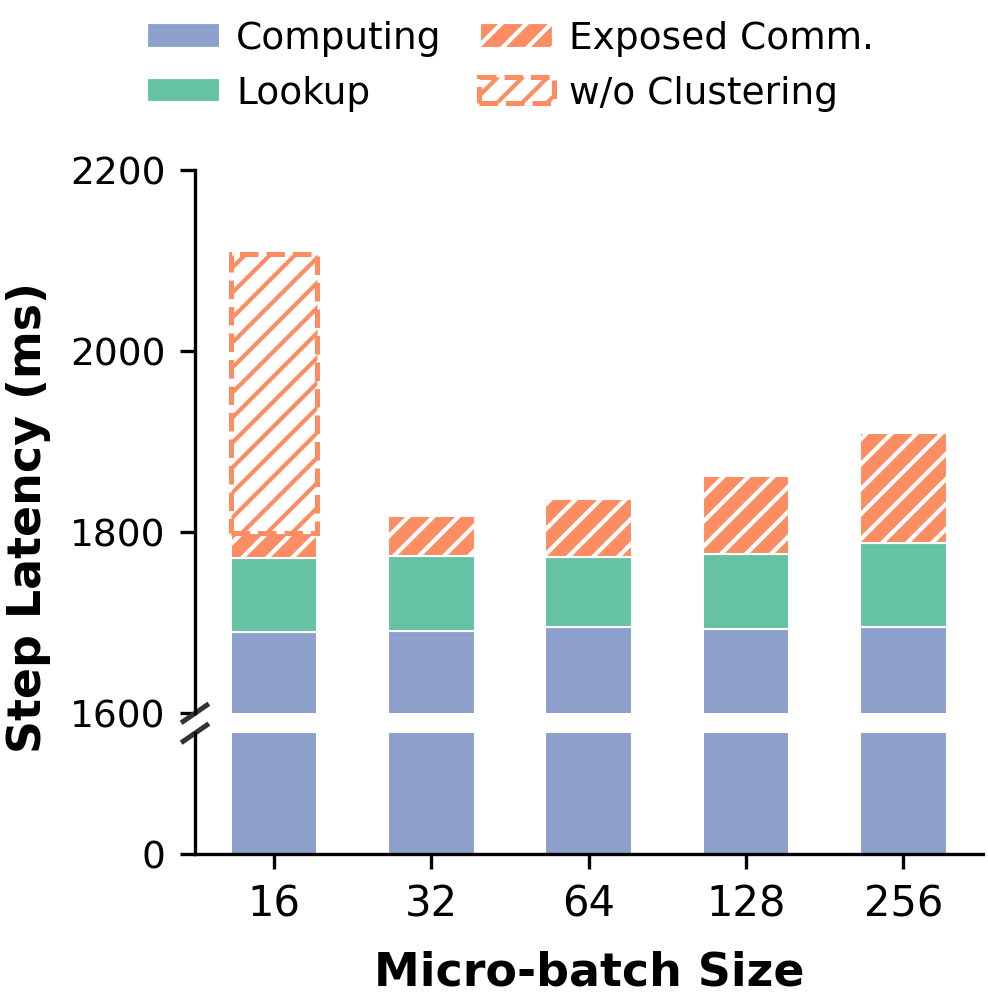}
        \caption{Latency Breakdown}
        \label{fig:mb_latency}
    \end{subfigure}
    \hfill
    \begin{subfigure}{0.49\columnwidth}
        \centering
        \includegraphics[width=\linewidth]{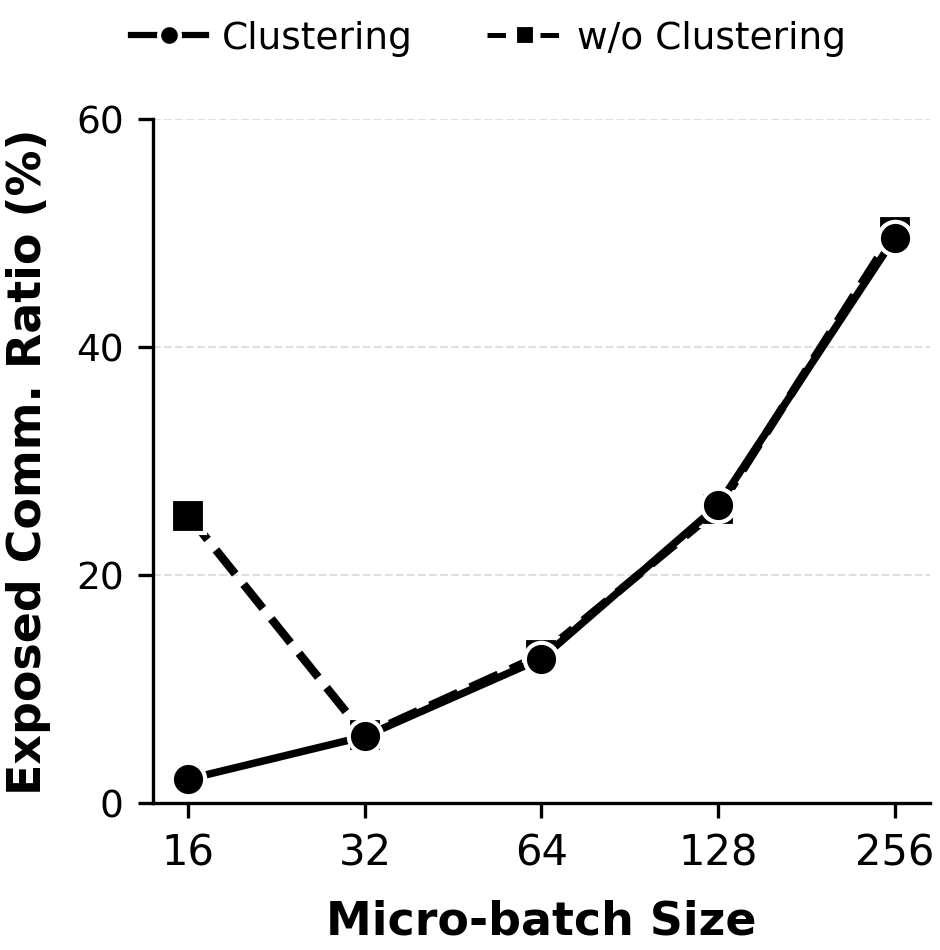}
        \caption{Exposed Comm. Ratio}
        \label{fig:mb_comm}
    \end{subfigure}
    
    \caption{Impact of micro-batch size on step latency and exposed comm. ratio under a constant batch size of 512.
    }
    \label{fig:micro_batch_analysis}
\end{figure}

\begin{figure*}
    \centering
    \includegraphics[width=\textwidth]{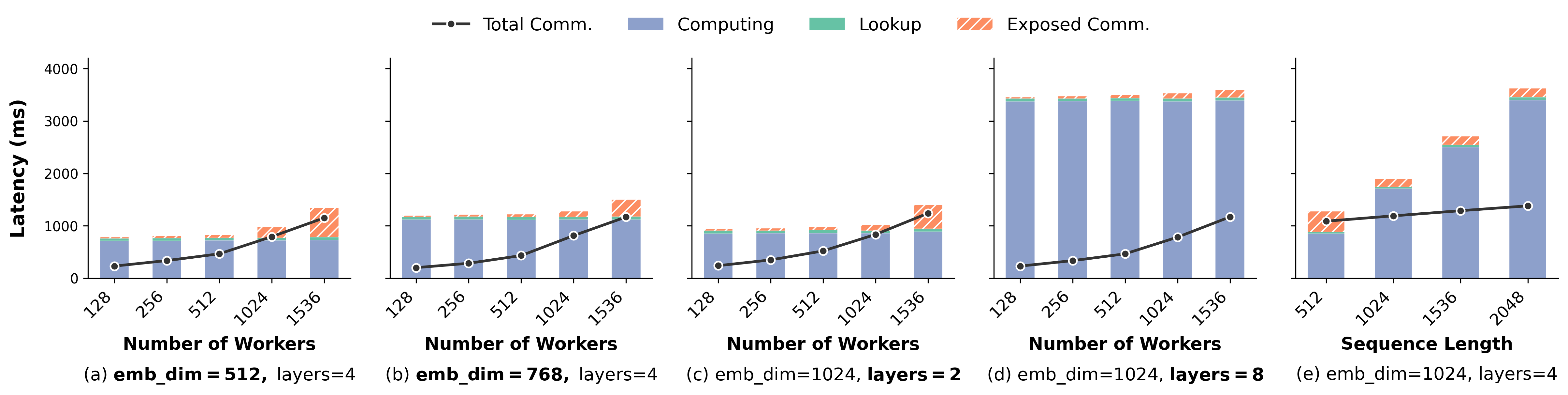} 
    \caption{Step latency breakdown for varying embedding dimensions, dense layers, and sequence lengths.
    }
    \label{fig:model_size}
    \vspace{2mm}
\end{figure*}

\subsection{RQ3: Scalability}
\label{sec:scalability}
\textbf{Scaling Factor.} We next investigate whether the benefit of NestPipe sustains as the training cluster scales from 128 to 1,536 workers.
Table~\ref{tab:scaling_performance} summarizes the throughput and the corresponding scaling factor.
We observe that the scaling factor of TorchRec and 2D-SP drops significantly to 44.34\% and 49.32\% when scaling to 1,536 workers. 
UniEmb maintains 89.21\% scaling efficiency up to 512 workers, but its advantage weakens at larger scales.
Our proposed framework preserves the scaling efficiency of 94.07\% even at the massive scale of 1,536 workers.
Fig.~\ref{fig:scaling} further illustrates the step latency breakdown. For baselines, the rapidly increasing lookup and communication overhead offset the additional computational resources, ultimately degrading training efficiency.
With the help of nested pipelining, NestPipe can reduce both lookup and communication overhead, which translates into 1.42-3.06$\times$ training speedup at different scales.

\textbf{Resource Utilization.} We further report the utilization ratio across varying cluster sizes in Fig.~\ref{fig:utilization}. 
Specifically, TorchRec and 2D-SP fall from 66.4\% and 70.5\% at 128 workers to 29.6\% and 34.8\% at 1,536 workers, respectively. 
Their strictly synchronous paradigms are insufficient at scale, leaving the expensive computing resources idle while waiting for data preprocessing and network transmission.
UniEmb improves resource utilization by pipeline parallelism, but still deteriorates when communication dominates. 
In comparison, NestPipe consistently maintains above 90\% hardware utilization via coordinated computation and communication stream scheduling, ensuring that the expensive computing resources remain highly active.
The above results imply that NestPipe achieves superior scalability in large-scale embedding training. 

\subsection{RQ4: Sensitivity Analysis}


\textbf{Impact of Micro-batch Size.} 
To verify the importance of key-centric sample clustering, we change the micro-batch size from 16 to 256 when training HSTU using the Industrial dataset on 512 NPU workers.
Fig.~\ref{fig:micro_batch_analysis} reveals that naively reducing micro-batch size does not translate to practical training speedup.
When the micro-batch size is reduced to 16, the physical All2All communication time inflates to 1,331.33ms, causing the actual exposed ratio 25.2\% (dashed line) to deviate from the theoretical $1/N$ bound.
In comparison, our sample clustering maximizes redundancy within small micro-batches and reduces repeated transmission. Therefore, the exposed communication payload decreases to 27.71ms.
\textbf{Impact of Model Scale.} We vary the embedding dimension ($emb\_dim\in\{512, 768, 1024\}$) and dense layers ($layers\in\{2, 4, 8\}$) to evaluate the performance of NestPipe across different workloads. As presented in Fig.~\ref{fig:model_size}(a)-(b), adjusting the embedding dimension (fixed at $layers=4$) increases both the communication and computation workload. 
The increased computation duration effectively absorbs the transmission latency, dropping the exposed comm. ratio to 12.9\%.
Fig.~\ref{fig:model_size}(c)-(d) show that scaling up the dense layers prolongs the computation duration, with no additional impact on lookup and communication payload. The 8-layer configuration provides the 3,393.12ms computation window that covers the communication time of 1,169.8ms, minimizing the exposed transmission latency to 146.23ms.
Even in the 2-layer configuration with narrow computation window (892.78ms), NestPipe still achieves exposed comm. ratio of 27.9\%. 

\textbf{Impact of Sequence Length.}
We quantify the impact of input sequence lengths varying from 512 to 2048 in Fig.~\ref{fig:model_size}(e). 
The dense computation time increases from 850.08ms to 3,399.96ms.
In comparison, the growth trend of actual lookup and communication time is moderated by key deduplication operation during embedding preprocessing.
Consequently, NestPipe can exploit the continuously expanding computation duration to cover the bulk of network transmission overhead. 
The communication latency exposed by NestPipe is strictly limited to 165.12ms at 2048 length, aligning with the theoretical exposed ratio. 
These results verify that our design is robust across compute-bound and communication-bound workloads. 

\begin{table}[t]
\centering
\caption{Integration of NestPipe and 2D-SP optimizations on 1,536 workers. QPS is reported in the scale of $10^5$, and the scaling factor is normalized to the 128-worker baseline.}
\label{tab:orthogonality_results}
\resizebox{\columnwidth}{!}{%
\begin{tabular}{l c c c c}
\toprule
\textbf{Method} & \textbf{\begin{tabular}[c]{@{}c@{}}Total Comm.\\ Latency (ms)\end{tabular}} & \textbf{\begin{tabular}[c]{@{}c@{}}Exposed Comm.\\ Latency (ms)\end{tabular}} & \textbf{\begin{tabular}[c]{@{}c@{}} QPS\end{tabular}} & \textbf{Scaling} \\ \midrule
TorchRec & 1207.85 & 1207.85 & 1.36 &  44.34\% \\
2D-SP & 438.36 & 438.36 & 1.60 &    49.32\%\\
NestPipe & 1185.60 & 154.23 & 4.14 &    94.07\%\\ \midrule
\cellcolor{gray!20}\textbf{NestPipe+2D-SP} & \cellcolor{gray!20}\textbf{452.34} & \cellcolor{gray!20}\textbf{55.64} & \cellcolor{gray!20}\textbf{4.32} & \cellcolor{gray!20}\textbf{97.17\%} \\ \bottomrule
\end{tabular}%
}
\end{table}

\subsection{RQ5: Complementary Optimization}
As previously highlighted, NestPipe is orthogonal to existing communication payload optimization methods like 2D-SP~\cite{zhang2025two}.
To validate the generality, we implement the NestPipe+2D-SP solution on the 1,536-worker cluster.
In this integrated setup, 2D-SP partitions workers into local groups to spatially reduce the physical All2All payload. 
Concurrently, NestPipe exploits the frozen window to temporally overlap this reduced communication time with the dense computation.
As summarized in Table~\ref{tab:orthogonality_results}, integrating 2D-SP aggressively compresses the raw communication latency to 452.34ms. Consequently, the $1/N$ exposed ratio proportionally shrinks to 55.64ms, ultimately boosting the throughput to 4.32$\times 10^5$ and achieving 97.17\% scaling factor on 1,536 workers.
Rather than competing with communication-reduction methods, NestPipe can amplify their benefits and unlock higher efficiency for large-scale embedding training.


\section{Conclusion}\label{sec:conclusion}
In this paper, we propose NestPipe, an efficient, consistent, and scalable framework for large-scale decentralized embedding training. 
Our framework tackles the lookup and communication bottlenecks through novel hierarchical sparse parallelism. DBP strategy constructs a staleness-free pipeline via lightweight buffer synchronization. FWP strategy leverages the parameter freezing window to decouple and hide communication latency, without breaking the semantics of standard synchronous training. Empirical evaluations on large-scale production clusters show that NestPipe provides superior throughput and scalability compared to existing methods.


\section*{Acknowledgment}
\noindent We gratefully acknowledge the authors at Huawei for their assistance in the Ascend hardware environment setup and system maintenance during the experiments.


\bibliographystyle{IEEEtran}
\bibliography{refs}











\end{document}